\documentclass{article}
\usepackage{ijcai18}

\usepackage{times}
\usepackage{xcolor}
\usepackage{soul}
\usepackage[utf8]{inputenc}
\usepackage[small]{caption}
\usepackage{amssymb}
\usepackage{amsmath}
\usepackage{amsthm}
\usepackage{wasysym}
\usepackage{graphicx}

\newcommand{\Q}{\mathbb{Q}}

\newcommand{\dualll}{\ensuremath{\mathit{dual}\mbox{-}\mathit{ll}}}
\renewcommand{\ll}{\ensuremath{\mathit{ll}}}
\newcommand{\dualpp}{\ensuremath{\mathit{dual}\mbox{-}\mathit{pp}}}
\newcommand{\pp}{\ensuremath{\mathit{pp}}}

\DeclareMathOperator{\Csp}{CSP}

\newcommand{\bA}{\mathcal{A}}
\newcommand{\bB}{\mathcal{B}}
\newcommand{\bC}{\mathcal{C}}
\newcommand{\bD}{\mathcal{D}}

\def\pa{\hbox{\textsf p}}
\def\pu{\hbox{\textsf p${}^{\smallsmile}$}}
\def\ma{\hbox{\textsf m}}
\def\mi{\hbox{\textsf m${}^{\smallsmile}$}}
\def\oa{\hbox{\textsf o}}
\def\oi{\hbox{\textsf o${}^{\smallsmile}$}}
\def\da{\hbox{\textsf d}}
\def\di{\hbox{\textsf d${}^{\smallsmile}$}}
\def\sa{\hbox{\textsf s}}
\def\si{\hbox{\textsf s${}^{\smallsmile}$}}
\def\fa{\hbox{\textsf f}}
\def\fu{\hbox{\textsf f${}^{\smallsmile}$}}

\newcommand{\xp}{\ensuremath{X^+}}
\newcommand{\xm}{\ensuremath{X^-}}
\newcommand{\yp}{\ensuremath{Y^+}}
\newcommand{\ym}{\ensuremath{Y^-}}

\newcommand{\spii}{\ \ }
\newcommand{\spiii}{\ \ \ }
\newcommand{\spiv}{\ \ \ \ }
\newcommand{\spv}{\ \ \ \ \ }

\newtheorem{lemma}{Lemma}
\newtheorem{theorem}{Theorem}
\newtheorem{corollary}{Corollary}

\title{
Classification transfer for qualitative reasoning problems
}

\author{
Manuel Bodirsky$^1$\thanks{Manuel Bodirsky has received funding from the ERC under the European Community's Seventh Framework Programme (Grant Agreement no. 681988, CSP-Infinity), and the DFG-funded project `Homogene Strukturen, Bedingungserf\"ullungsprobleme, und topologische Klone' (Project number 622397).}, 
Peter Jonsson$^2$\thanks{Partially supported by the {\em Swedish Research Council}
(VR) under Grant 2017-04112.},
Barnaby Martin$^3$, 
Antoine Mottet$^1$\thanks{Supported by the DFG programme QuantLA.} 
\\ 
$^1$ Institut f\"ur Algebra, TU Dresden, Germany \\
$^2$ \mbox{Dept.} of Computer and Information Science, 
Link\"opings Universitet, Sweden \\
$^3$ \mbox{Dept.} of Computer Science, Durham University, UK \\
}
\begin{document}

\maketitle

\begin{abstract}
%We study formalisms for temporal and spatial reasoning in the modern context of Constraint Satisfaction Problems (CSPs). Questions on the complexity of their subclasses can now be solved using existing results via the powerful use of primitive positive (pp) interpretations and pp-homotopy. In the case of the Rectangle Algebra we answer in the affirmative the old open question as to whether ORD-Horn is a maximally tractable subset among the (disjunctive, binary) relations. We then generalise our results for the Rectangle Algebra to the $r$-dimensional Block Algebra.
We study formalisms for temporal and spatial reasoning in the modern, algebraic and model-theoretic, context of infinite-domain Constraint Satisfaction Problems (CSPs). We show how questions
on the complexity of their subclasses can be solved
using existing results via the powerful use of primitive positive (pp) interpretations and pp-homotopy.
We demonstrate the methodology by giving a full complexity
classification of all constraint languages that are
first-order definable in Allen's Interval Algebra and contain
the basic relations $(\sa)$ and $(\fa)$. In the case of the Rectangle Algebra
we answer in the affirmative the old open question as to whether
ORD-Horn is a maximally tractable subset among
the (disjunctive, binary) relations. We then generalise our results for the Rectangle Algebra to the
$r$-dimensional Block Algebra.

\end{abstract}

\section{Introduction}
A classical line of enquiry within Artificial Intelligence (AI) considers an agent's ability to reason about time and space, and a wide variety of formalisms for temporal and spatial reasoning are surveyed in \cite{TemporalSpatialSurvey}. Famous such examples are the Interval Algebra of \cite{Allen}, the Region Connection Calculus of \cite{Randell92aspatial} and the Cardinal Direction Calculus of \cite{LigozatCDC}. The central computational problem, designated the \emph{fundamental task of Qualitative Spatial and Temporal Reasoning} in \cite{NebelRenzSurvey}, has to do with the consistency of some partial information given in one of these formalisms. In his article on Interval Algebras, Hirsch [\citeyear{Hirsch}] outlines \emph{The Really Big Complexity Problem} associated to some subclass of a formalism, that is the (computational) complexity of its consistency problem. The complexity of the entire formalism is usually too blunt a measure, for example the Point Algebra \cite{VilainK86} has long been known to be tractable, while the Interval Algebra is NP-complete \cite{NebelFirst}. Much more interesting is a map of the complexity-theoretic landscape in terms of subclasses of the formalism. In the case of the Interval Algebra the natural classes of concern are subsets of disjunctions of the basic relations, and there are then \emph{a priori} $2^{2^{13}}$ such classes requiring a complexity classification. The tower representation produces unfortunate typography and the relevance of this number is only in its forbidding largeness, so let us replace it henceforth with a $\frownie$. This classification problem for the Interval Algebra tapped into a rich vein inspiring the popular paper of Nebel and B\"urckert [\citeyear{Nebel}], who identified a class of relations within the $\frownie$-many that is \emph{maximally tractable}, in the sense that its consistency problem is in P, yet if one adds any other relation from the $\log(\frownie)$-many the problem becomes NP-complete. This maximal class was named \emph{ORD-Horn} and owes its tractability to a variant of the local consistency method by which Horn Satisfiability is resolved in P (a very fast algorithm for the latter is given in \cite{HornSATLinear}).
% Indeed, such classes based on Horn's restriction \cite{horn1951} are ubiquitous in this field.
 A complete classification showing the 18 maximal tractable classes among the $\frownie$-many was finally given in \cite{KrokhinAllen}.

A \emph{constraint satisfaction problem} (CSP) is a computational problem in which the input consists of a finite set of variables and a finite set of constraints, and where the question is whether there exists a mapping from the variables to some fixed domain such that all the constraints are satisfied. The set of relations that is allowed to formulate the constraints in the input is called the \emph{constraint language}. 
It is well-known that the consistency problem for some subclass of a temporal or spatial formalism is an example of a CSP with a constraint language over an infinite domain. Indeed, the modern and algebraic study of CSPs has played an important role in answering some of these classical questions from AI, \mbox{e.g.} for the Interval Algebra in \cite{KrokhinAllen} (though other times, \mbox{e.g.} for the RCC-5, such questions were answered by exhaustive and computer-driven search through roughly $\frownie$-many cases \cite{RCC5JD}). 

By now the literature on infinite-domain CSPs is beginning to mature. Much of the modern work has been motivated by questions originating in AI. The constraint languages are chosen according to what mathematical methods they might be amenable to, rather than their being idiosyncratic to spatio-temporal reasoning. These are often ordinary structures of arithmetic, but may be more elaborate mathematical constructions (such as the Fra\"iss\'e limits used for RCC-8 in \cite{BodirskyW11}). This need not be a hindrance, after all the Interval Algebra can be embedded in the real line and the majority of the spatio-temporal formalisms (unsurprisingly) have a natural interpretation in some Euclidean space. The notion of definability is usually in first-order logic rather than simple disjunctions of atomic relations. Outstanding work in this direction includes the complexity classifications for the \emph{temporal CSPs}, fo-definable in $(\mathbb{Q};<)$ \cite{tcsps-journal} and \emph{discrete temporal CSPs}, fo-definable in $(\mathbb{Z};<)$ \cite{dCSPs3}.
%, together with certain extending linear program feasibility \cite{essentially-convex} and integer program feasibility \cite{JonssonLoow}. 
In these examples, a complexity-theoretic dichotomy, between P and NP-complete, is observable across the different constraint languages. Such a dichotomy is reminiscent of the case for finite-domain CSPs as detailed in the recently-proved \cite{FVproofBulatov,FVproofZhuk} Feder-Vardi Conjecture \cite{FederVardi}. While infinite-domain CSPs may have much higher complexity, \mbox{e.g.} be undecidable \cite{BodirskyGrohe}, in well-behaved cases many of the methods developed for the finite-domain are applicable. Thus, the technology exists to answer many instances of the Really Big Complexity Question. Some of the tools and techniques reappear over and again, for example tractability based on Horn's restriction \cite{horn1951} (see \cite{HornOrFull}). Yet, the proofs appearing for the Interval Algebra in \cite{KrokhinAllen} and the temporal CSPs of \cite{tcsps-journal} are quite distinct and often rely on local ad-hoc constructions. The observation that the Interval Algebra can be embedded in the (real or) rational line belies a new hope. Can a small number of existing classifications allow the simple derivation of other classifications? Could the classifications for the Interval Algebra and temporal CSPs be obtained, in some form, as a corollary of the other? The answer to this question is (a mildly-qualified) yes, so long as we can assume the basic relation $\texttt{m}$ to be in the language defined in the Interval Algebra, and $<$ to be in the language defined in $(\mathbb{Q};<)$. More interesting is to use existing classifications to simply derive solutions to open problems. In this vein, we will look at the Rectangle Algebra of \cite{Giisgen89,MukerjeeJ90} and its generalisation to the $r$-dimensional Block Algebra of \cite{BalbianiCondottaCerro-journal}.
%, the Cardinal Direction Calculus of \cite{LigozatCDC} and the Directed Interval Algebra of \cite{Renz}.

ORD-Horn is known to be a tractable fragment not only for the Interval Algebra but also for its generalisations to the (2-dimensional) Rectangle Algebra and ($r$-dimensional) Block Algebra \cite{BalbianiCondottaCerro-journal}. In that paper it is noted that:  ``The problem of the maximality of this tractable subset [ORD-Horn] remains an open problem. Usually to prove the maximality of a fragment of a relational algebra an extensive machine-generated analysis is used. Because of the huge size ... we cannot proceed in the same way.'' Our method is able to answer that ORD-Horn is indeed maximally tractable for the Rectangle Algebra as well as the $r$-dimensional Block Algebra, without a computational search, but via knowledge of maximal classes from \cite{tcsps-journal}.

This paper is structured as follows. We introduce the preliminaries as well as the key notions of interpretations and homotopy in Section~\ref{sec:2}. Our principal results on classification transfer are given in Section~\ref{sec:transfer}, while we address maximal tractability of ORD-Horn in Section~\ref{sec:ORD-Horn}. %We conclude in Section~\ref{sec:final} with some final remarks.

\section{Preliminaries}
\label{sec:2}

Let $\bB$ be a finite (relational) structure, that is a domain $B$ embellished with a finite set of relations on that domain. The sequence of arities of these relations (together with their names) constitutes the \emph{signature} of $\bB$. Of course, one may equivalently think of $\bB$ as a finite set of relations over the same domain. We may now give a formal definition of the \emph{constraint satisfaction problem} when it
is parameterised by a set of relations, equivalently a (relational) structure. 
The problem CSP$(\bB)$, where $\bB$ is known as the \emph{constraint language}, is defined as follows:

\medskip

\noindent
{\em Instance:} A set $V$ of variables and a set $C$ of {\em
  constraints} of the form
$R(v_1,\ldots,v_k)$, where $k$ is the arity of $R$, $v_1,\ldots,v_k \in V$ and $R$ a relation of $\bB$.

\noindent
{\em Question:} Is there a function $f:V \rightarrow B$ such that
$(f(v_1),\ldots,f(v_k)) \in R$ for every $R(v_1,\ldots,v_k) \in C$?
\medskip

%First-order formulas $\phi$ over
%the signature $\tau$ (or, in short, $\tau$-formulas) are inductively
%defined using the logical symbols of universal
%and existential quantification ($\forall$ and $\exists$), disjunction ($\vee$), conjunction ($\wedge$), negation ($\neg$),
%equality, bracketing, variable
%symbols and the symbols from $\tau$.
%The semantics of a first-order formula over some $\tau$-structure is
%defined in the usual Tarskian style.

First-order formulas $\phi$ over
the signature $\tau$ (or, in short, $\tau$-formulas) are inductively
defined using the logical symbols of universal
and existential quantification ($\forall$ and $\exists$), disjunction ($\vee$), conjunction ($\wedge$), negation ($\neg$),
equality, bracketing, variable
symbols and the symbols from $\tau$.
The semantics of a first-order formula over some $\tau$-structure is
defined in the usual Tarskian style. When $\phi$ is a formula without free variables (a \emph{sentence}), we write $\bB \models \phi$ to indicate that $\phi$ is true on $\bB$.

One can use first-order formulas to define relations
over a given structure $\bB$: for a formula $\phi(x_1,...,x_k)$ the corresponding relation $R$
is the set of all $k$-tuples $(t_1,...,t_k) \in B^k$ such that $\phi(t_1,...,t_k)$
is true in $\bB$. In this case we say that $R$ is first-order definable over $\bB$.
Note that our definitions are always parameter-free, \mbox{i.e.} we do not allow
the use of domain elements in them. A structure $\bA$ is said to be \emph{first-order definable} over $\bB$ when all of its relations are first-order definable in $\bB$, it is \emph{first-order expansion} of $\bB$ if it  further includes all of the relations of $\bB$.

The fragment of first-order logic called \emph{primitive positive} (pp) is that which involves only existential quantification ($\exists$) and conjunction ($\wedge$). In particular, primitive positive formulas may involve equality, but no negation. An alternative definition of CSP$(\bB)$, logspace equivalent to that which we gave, is furnished by the model-checking problem for primitive positive logic, on the fixed model $\bB$. In this formulation it is easy to see the following folkloric result.
\begin{lemma}[Theorem 3.4 in \cite{Jeavons}]
If a structure $\bA$ has a finite number of relations, all of which are pp-definable in $\bB$, then there is a polynomial-time reduction from CSP$(\bA)$ to CSP$(\bB)$.
\label{lem:folklore}
\end{lemma}
%A \emph{homomorphism} from $\bA$ to $\bB$ over the same signature is a map $h$ from $A$ to $B$ so that for each tuple $(x_1,\ldots,x_k)$ in a relations $R$ of $\bA$, $(h(x_1),\ldots,h(x_k)$ is in the corresponding relation $R$ of $\bB$. A bijective homomorphism whose inverse is a homomorphism is called an \emph{isomorphism} and an isomorphism between a structure and itself is called an \emph{automorphism}. 
%We will only consider finite $\Gamma$ and infinite $D$ in this 
%paper. 
%Given an instance $I$ of CSP$(\Gamma)$ we write $||I||$ for the
%number of bits required to represent $I$. 

Suppose now ${\cal B}$ is a structure with a finite number of relations, each of which
has arity $m$ (this set is referred to as the {\em basic relations}).
Define ${\cal B}^{\vee =}$ to contain every $m$-ary relation $R$ such
that $R({\overline x})$ holds if and only if $B_{1}({\overline x}) \vee \dots \vee B_{t}({\overline x})$ holds,
where $\vee$ is the disjunction operator, 
$\{B_1,\dots,B_t\} \subseteq {\cal B}$, and ${\overline x}=(x_1,\ldots,x_m)$
is a variable vector. 
Clearly, ${\cal B}^{\vee =}$ is a constraint language with a finite set of relations.
When ${\cal B}$ contains binary relations, we write 
$x(B_1 \: B_2 \: \dots \: B_t)y$ instead of
$B_{1}(x,y) \vee B_2(x,y) \vee \dots \vee B_{t}(x,y)$.

%We note that the intersection of two relations $R,S \in {\cal B}^{\vee =}$ is
%itself a member of ${\cal B}^{\vee =}$: 
%if $R({\bf x}) \Leftrightarrow
%B_{i_1}({\bf x}) \vee \dots \vee B_{i_t}({\bf x})$ and
%$S({\bf x}) \Leftrightarrow
%B_{j_1}({\bf x}) \vee \dots \vee B_{j_u}({\bf x})$, then
%\[(R \cap S)({\bf x}) \Leftrightarrow \bigvee_{m \in \{i_1,\dots,i_t\} \cap \{j_1,\dots,j_u\}} B_m({\bf x}).\]

\subsection{The Interval and Block Algebras}

\begin{table}[h]
\begin{center}
     \begin{tabular}{|ll|l|l|}\hline
       Basic relation \hspace*{3mm} & \hspace*{8mm} & Example & Endpoints \hspace*{2mm}\\ 
\hline\hline
      $X$ precedes       $Y$ & $\pa$  & \texttt{XXX\spv}    & $\xp < \ym$ 
\\ \cline{1-2}
      $Y$ preceded by         $X$ & $\pu$  & \texttt{\spv YYY}   & \\ 
\hline
      $X$ meets $Y$ & $\ma$   & \texttt{XXXX\spiv}  & $\xp=\ym$ 
\\ \cline{1-2}
      $Y$ is met by        $X$ & $\mi$  & \texttt{\spiv YYYY} & \\ 
\hline
      $X$ overlaps      $Y$ & $\oa$   & \texttt{XXXX\spii} & $\xm<\ym \: \wedge$\\ 
& & \texttt{\spii YYYY} & $\ym <\xp \: \wedge$ \\ \cline{1-2}
      $Y$ overlapped by $X$ & $\oi$  & & $\xp <\yp \: \wedge$
 \\ \hline
      $X$ during        $Y$ & $\da$   & \texttt{\spiii XX \spii} & $\xm>\ym \: \wedge$ \\ \cline{1-2}
      $Y$ includes      $X$ & $\di$  & \texttt{YYYYYY}   & $\xp<\yp$ \\ 
\hline
      $X$ starts        $Y$ & $\sa$   & \texttt{XXX\spiv}  & $\xm=\ym \: \wedge$
\\ \cline{1-2}
      $Y$ started by $X$ & $\si$  & \texttt{YYYYYY}   & $\xp<\yp$ \\ 
\hline
      $X$ finishes      $Y$ & $\fa$   & \texttt{\spiii XXX} & $\xp=\yp \: \wedge$
\\ \cline{1-2}
      $Y$ finished by $X$ & $\fu$  & \texttt{YYYYYY}   & $\xm>\ym$ \\ 
\hline
      $X$ equals        $Y$ & $\equiv$ & \texttt{XXXX}      & $\xm=\ym \: \wedge$ \\ 
                                   &          & \texttt{YYYY}      & $\xp=\yp$ \\ 
\hline
    \end{tabular}
\end{center}
\caption{Basic relations in the interval algebra.} \label{tb:allen-basic-defs}
\end{table}

The {\em interval algebra}~\cite{Allen} (IA) is a formalism that is 
both well-known and well-studied in AI. 
The variable domain is 
\[{\mathbb I}=\{\{x \in {\mathbb Q} \; | \; a \leq x \leq b\} \; | \; a,b \in {\mathbb Q} \; {\rm and} \; a < b\},\]
i.e. the variable domain consists of all closed intervals $[a,b]$ of 
rational numbers.
If $I=[a,b] \in \mathbb I$, then we write $I^-$ for $a$ and $I^+$ for $b$.
The basic relations are the 13 relations defined in Table~\ref{tb:allen-basic-defs}. We let ${\cal IA}$ denote the structure that is the set of basic interval relations over $\mathbb{I}$.
Clearly, the 8192 relations
of the IA are the contents of the set ${\cal IA}^{\vee =}$. Among them is the binary relation $\top$ which holds for all pairs of intervals.

Given a sequence $S=(s_1,\dots,s_p)$, we let $S[i] = s_i$,
$1 \leq i \leq p$.
Let $p \geq 1$ be an integer. We will now define the 
$p$-{\em dimensional block algebra} (BA$_p$). 
The domain is ${\mathbb I}^{p}$. 
Given a sequence $(\mathtt{r}_1,\dots,\mathtt{r}_p)$ where $\mathtt{r}_1,\dots,\mathtt{r}_p$ are relations of ${\cal IA}$,
we define the binary relation $B_{(\mathtt{r}_1,\dots,\mathtt{r}_p)}=\{(X,Y) \in ({\mathbb I}^{p})^2  \; | \; X[i] (\mathtt{r}_i) Y[i], \; 1 \leq i \leq p\}$. 
The basic relations in BA$_p$ constitute the structure ${\cal BA}_p:=\{B_{(\mathtt{r}_1,\dots, \mathtt{r}_p)} \; | \; \mathtt{r}_1,\dots,\mathtt{r}_p \mbox{ from } {\cal IA}\}$ and the relations of
BA$_p$ are the members of ${\cal BA}_p^{\vee =}$.
We note that BA$_1$ is the interval algebra and that BA$_2$ is often referred
to as the
{\em rectangle algebra} which we denote as RA with associated structure $\mathcal{RA}$. Let us note here that, for each $p$, ${\cal BA}_p^{\vee =}$ coincides with the set of binary relations first-order definable over ${\cal BA}_p$, since this latter structure admits quantifier elimination.
%\begin{array}{ll}
%\exists Z_1,Z_2,Z_3,Z_4 & (\equiv,s)(X,Z_1) \wedge (\equiv,f)(Z_1,Z_2) \wedge \\
%& (\equiv,s)(Y,Z_3) \wedge (\equiv,f)(Z_3,Z_4) \wedge \\
%& (x,=)(Z_2,Z_4).
%\end{array}

\subsection{Interpretations}

A \emph{first-order interpretation} $I$ of a structure $\bB$ over signature $\tau$, in a structure $\bA$ over signature $\sigma$, is a triple $(k,\delta,g)$ where $k \in \mathbb{N}$ is the \emph{dimension} of the reduction, $\delta$ is a first-order definable subset of $\bA^k$ and $g$ is a surjection from $\delta$ to $B$. We further require that, for every $s$-ary relation $S$ in $\bB$, as well as for equality, there is a first-order definable $\phi(x^1_1,\ldots,x^k_1,\ldots,x^1_s,\ldots,x^k_s)$ over $\sigma$ so that, for all $(x^1_1,\ldots,x^k_1), \ldots,(x^1_s,\ldots,x^k_s)  \in \delta$:
\[
\begin{array}{c}
 \bA \models \phi(x^1_1,\ldots,x^k_1,\ldots,x^1_s,\ldots,x^k_s) \ \Leftrightarrow \\
  \bB \models S(g(x^1_1,\dots,x^k_1),\ldots,g(x^1_s,\dots,x^k_s)).
\end{array}
  \]
When $\phi$ is the first-order definition of $R$ in $\bB$, then
$S$ is the relation defined by $\phi_{I}$ in $\bA$, where $\phi_I$ is built from $\phi$ by substituting each atomic formula in the first-order definition by its translation under $I$.

Let $I$ be an $i$-ary interpretation of $\bB$ in $\bA$ denoted $(i,\delta,g)$ and $J$ be a $j$-ary interpretation of $\bC$ in $\bB$ denoted $(j,\epsilon,h)$. Then, define $J \circ (I,\ldots,I)$ ($I$ is written $j$ times), 
%also paraphrased ``$J \circ I$'',
 to be the interpretation $(ji,\gamma, f)$ defined as follows. Let $\gamma(x^1_1,\ldots,x^i_1,\ldots,x^1_j,\ldots,x^i_j)$ be defined by the conjunction of $\delta(x^1_1,\ldots,x^i_1) \wedge \ldots \wedge \delta(x^1_j,\ldots,x^i_j)$ and $\epsilon'(x^1_1,\ldots,x^i_1,\ldots,x^1_j,\ldots,x^i_j)$, where the latter is built from $\epsilon(x_1,\ldots,x_j)$ by substituting atoms (over the signature of $\bB$) by their first-order definition in $\bA$, as exists from $I$, in the obvious fashion. Note that each $x_\lambda$ becomes $i$ new variables $x^1_\lambda,\ldots,x^i_\lambda$ under the interpretation $I$.
Set $f$ to be $h(g(x^1_1,\ldots,x^i_1), \ldots,g(x^1_j,\ldots,x^i_j))$. In other words, a first-order definition of some $k$-ary atomic relation $R$ of $\bC$ in $\bA$ is furnished by taking the $ik$-ary first-order definition of $R$ in $\bB$ guaranteed by $J$ and substituting atoms in that by their first-order definition in $\bA$ as guaranteed by $I$.
Thus we can say that interpretations are transitive. We are most interested in \emph{pp-interpretations} when the first-order definitions are indeed primitive positive. Primitive positive interpretations, like general (first-order) interpretations, are also preserved under composition (note this is not true, e.g., for existential interpretations). We extend (pp-)interpretation of relations to atomic formulas in the obvious fashion.

By way of example, let us consider the binary interpretation of $\mathcal{IA}^{\vee =}$ in $(\Q;<)$, given by $I=(2,x<y,h)$, in which $h(X^-,X^+)=[X^-,X^+]$ maps an ascending pair of rationals to the interval whose endpoints they specify. It is easy to verify that each of the relations of $\mathcal{IA}^{\vee =}$ is first-order definable in $(\Q;<)$.
% As an example, $[X^-,X^+](\sa)[Y^-,Y^+]$ is given by $X^-=Y^- \wedge X^+<Y^+$. 
Indeed, each of the (basic) relations of $\mathcal{IA}$ is pp-definable in $(\Q;<)$ (see the third column of Table~\ref{tb:allen-basic-defs}). Thus $\mathcal{IA}$ is even  pp-interpretable in $(\Q;<)$.

Two interpretations of $\bC$ in $\bD$
with coordinate maps $h_1$ and $h_2$ are called 
\emph{homotopic}\footnote{We are following the terminology from~\cite{AhlbrandtZiegler}.}
if the relation $\{(\bar x,\bar y) \; | \; h_1(\bar x) = h_2(\bar y) \}$
is first-order definable in $\bD$. If this relation is even primitive positive
definable in $\bD$, we say that the two interpretations are \emph{pp-homotopic}.
The \emph{identity interpretation} of a $\tau$-structure $\bC$
 is the interpretation $I=(1,\top,h)$
of $\bC$ in $\bC$ whose coordinate map $h$ is the identity
(note that the identity interpretation is primitive positive). 

Two structures are said to be \emph{mutually pp-interpretable} when there is a pp-interpretation of each in the other. When these further satisfy a pp-homotopy condition then we will have enough for classification transfer.

Most important for our results is a novel form of composition of interpretations that we will now introduce. We previously saw how to compose two interpretations $J$ and $I$ through $J\circ(I,\ldots,I)$ but imagine now we have not one $I$ but several. So suppose instead that we have $j$ different $i$-ary interpretations 
$I_1:=(i,\delta_1,g_1)$, \ldots, $I_j:=(i,\delta_j,g_j)$, of $\bB$ in $\bA$; together with $J:=(j,\epsilon,h)$, a $j$-ary interpretation of $\bC$ in $\bB$.
Then we can define also $J \circ (I_1,\ldots,I_j)$ in the following fashion, which is indeed consistent with our previous definition of $J \circ (I,\ldots,I)$ in the case $I=I_1=\cdots=I_j$.
Consider the partial $ij$-ary surjection $h(g_1,\ldots,g_j)$, defined through $h(g_1(x_1,\ldots,x_i),\ldots,g_j(x_{ij-i+1},\ldots,x_{ij}))$, from $A^{ij}$ to $C$. Note that it is defined precisely on $ij$-tuples satisfying $\delta_1(x_1,\ldots,x_i) \wedge \ldots \wedge \delta_j(x_{ij-i+1},\ldots,x_{ij}) \wedge \epsilon(g_1(x_1,\ldots,x_i),\ldots,g_j(x_{ij-i+1},\ldots,x_{ij}))$. Indeed, interpretations may be seen as partial surjections and partial surjections are closed under the composition we have just seen. The partial surjection $h(g_1,\ldots,g_j)$ bestows a natural map from atomic relations of $\bC$ to relations of $\bA$ but there is no guarantee of nice (pp-, or even fo-) definability in $\bA$. Even the definition we have given for the domain of $h(g_1,\ldots,g_j)$ is not syntactic in the sense of fo-definability. However, let us now consider the case $\bC=\bA$ in which $h(g_1,\ldots,g_j)$ is itself pp-definable, in the sense that the $ij+1$-ary relation on $\bA$ given by $h(g_1(x_1,\ldots,x_i),\ldots,g_j(x_{ij-i+1},\ldots,x_{ij}))=y$ is pp-definable. This is exactly $J \circ (I_1,\ldots,I_j)$ being pp-homotopic with the identity interpretation (what we may paraphrase in future as simply ``pp-homotopy'', when the context is clear. In this case, it is clear that domain and relations, for $J \circ (I_1,\ldots,I_j)$, are pp-definable. For example, the domain formula is $\exists y \ h(g_1(x_1,\ldots,x_i),\ldots,g_j(x_{ij-i+1},\ldots,x_{ij}))=y$ and a $k$-ary relation $R$ of $\bA$ gives rise to its interpreting $ijk$-ary relation of $\bA$ given by 
\[
\begin{array}{l}
\exists y^1,\ldots,y^k R(y^1,\ldots,y^k) \wedge \\
h(g_1(x^1_1,\ldots,x^1_i),\ldots,g_j(x^1_{ij-i+1},\ldots,x^1_{ij}))=y^1 \wedge \\
\vdots \\
h(g_1(x^k_1,\ldots,x^k_i),\ldots,g_j(x^k_{ij-i+1},\ldots,x^k_{ij}))=y^k \\
\end{array}
\]
where here subscripts are read lexicographically before superscripts.

As an example, let us consider the binary interpretation $J$ of $\mathcal{RA}$ in $\mathcal{IA}$, given by $J=(2,\top,h)$, in which $h(X,Y)=(X,Y)$ maps an ordered pair of intervals to the rectangle they specify in the plane. Consider further the two unary interpretations $I_1$ and $I_2$, given by $(1,\top,p_1)$ and $I_2:=(1,\top,p_2)$, respectively, where $p_1$ and $p_2$ are the binary projections to horizontal and vertical axes. According to our definitions, we are at liberty to consider any of the interpretations $J \circ (I_1, I_1)$, $J \circ (I_2, I_2)$ and $J \circ (I_1, I_2)$. The important observation for us is that only the third of these is pp-homotopic to the identity interpretation. The first two are not even fo-homotopic to the identity. All the results in this paper require a similar form.

\section{Classification Transfer}
\label{sec:transfer}
Let $\bC$ be a structure with finite relational signature. By the
\emph{classification project for $\bC$} we mean a complexity
classification for $\Csp(\bB)$ for all first-order 
expansions $\bB$ of $\bC$ that have finite relational signature. For instance, the classification project for the random graph is treated in \cite{BodPin-Schaefer} and the classification project for
$(\Q;<)$ is treated in \cite{tcsps-journal}.

Sometimes, it is possible to derive the (complexity) classification
project for a structure $\bC$ from the (complexity) classification project 
for another structure $\bD$. The following lemma is our principal tool and gives us the method by which we can demonstrate classification transfer.

\begin{lemma}
\label{lem:transfer-new}
Suppose $\bD$ has $j$ $i$-ary primitive positive interpretations $I_1,\ldots,I_j$ in $\bC$,
and $\bC$ has a $j$-ary primitive positive interpretation $J$ in $\bD$ such that
 $J \circ (I_1,\ldots,I_j)$ is pp-homotopic to the identity interpretation of $\bC$.
Then for every first-order expansion $\bC'$ of $\bC$
there is a first-order expansion $\bD'$ of $\bD$
such that $\bC'$ and $\bD'$ are mutually pp-interpretable.
\end{lemma}
\begin{proof}
Let $I_1 = (i,U_1,g_1)$, \ldots,  $I_j = (i,U_j,g_j)$ and $J=(j,V,h)$ be the primitive positive interpretations from the statement,
and let $\bC'$ be 
a first-order expansion of $\bC$.
Then we set $\bD'$ to be the expansion of $\bD$
that contains for every $k$-ary atomic formula (derived from atomic relation $R$) in the signature of $\bC'$ 
the $jk$-ary relation $S$ defined as follows. 
When $\phi$ is the first-order definition of $R$ in $\bC$, then
$S$ is the relation defined by $\phi_{J}$ in $\bD$ (recall that $\phi_J$ is built from $\phi$ by substituting each atomic formula in the first-order definition by its translation under $J$).

We claim that $\bC'$ is primitive positive interpretable in $\bD'$, indeed by $(j,V,h)$. First note that $V$ is primitive positive definable
in $\bD'$ since $\bD'$ is an expansion of $\bD$.
An atomic formula $\psi$ with free variables $x_1,\dots,x_k$
in the signature of $\bC'$ can be interpreted in $\bD'$ as follows. We replace the relation
symbol in $\psi$ by its definition in $\bC$, and obtain
a formula $\phi$ in the language of $\bC$.
Let $S$ be the symbol in the language of $\bD'$ for the relation
defined by $\phi_{J}(x_1^1,\dots,x_1^j,\dots,x^1_k,\dots,x^j_k)$
 over $\bD'$. Then indeed $S(x_1^1,\dots,x_1^j,\dots,x^1_k,\dots,x^j_k)$ is a defining formula for $\psi$, because 
\[
\begin{array}{c}
\bC' \models \psi(h(a_1^1,\dots,a_1^j),\dots,h(a^1_k,\dots,a^j_k)) \Leftrightarrow \\ 
\bD' \models S(a_1^1,\dots,a_1^j,\dots,a^1_k,\dots,a^j_k)
\end{array}
\]
for all $a_1,\dots,a_k \in V$.

Conversely, we claim that $\bD'$ has each of the primitive positive interpretations  $(i,U_1,g_1)$, \ldots,  $(i,U_j,g_j)$  in $\bC'$. 
As before, each of $U_1,\ldots,U_j$ is primitive positive definable in $\bC'$ since 
$\bC'$ is an expansion of $\bC$. 
Let $\phi$ be a $k$-ary atomic formula in the (relational) signature of $\bD'$. If $\phi$ is already in the signature of $\bD$, then there is a primitive positive interpreting formula in $\bC$ and therefore also in $\bC'$.
Otherwise, by the definition of $\bD'$, the relation
symbol in $\phi$ has arity $jk'$, and has been introduced for a ($k\leq$) $k'$-ary relation $R$ from $\bC'$ (if $k<k'$ some coordinates were identified in $R$). We have to find a defining formula
having $ijk$ free variables, but we will build one with $ijk'$ variables in which some have to be identified in line with our previous comment. 
Let $\theta(x_0,x_{1},\dots,x_{ij})$ be the primitive positive formula of
arity $ij+1$ that shows 
that $h(g_1(x_{1},\ldots,x_i),\dots,g_j(x_{ij-i+1},\ldots,x_{ij})) = x_0$ is primitive positive definable in $\bC$.
Then the defining formula for the atomic formula $\phi(x^1,\ldots,x^{k'})$, where some of these variables are identified, is
$\exists x^1,\dots,x^{k'} \; 
\big(R(x^1,\dots,x^{k'}) \; \wedge \;  \bigwedge_{\ell=1}^{k'} \theta(x^\ell,x^\ell_{1},\dots,x^\ell_{ij})\big )$,
where here subscripts are read lexicographically before superscripts. To see that this defines $\phi$, it is sufficient to note that $\theta$ defines equality.
\end{proof}
%At first sight, this new version of Lemma~\ref{lem:transfer} is unnecessarily more complicated that the original. However, the point is that $J \circ (I_i,\ldots,I_i)$ may not be pp-homotopic, for any $i \in [m]$, while  $J \circ (I_1,\ldots,I_m)$ is easily seen to be pp-homotopic.
%The following corollary is the cornerstone of classification transfer. It follows immediately from Lemmas~\ref{lem:folklore} and \ref{lem:transfer-new}.
\begin{corollary}
\label{cor:transfer-new}
Suppose $\bD$ has $j$ $i$-ary primitive positive interpretations $I_1,\ldots,I_j$ in $\bC$,
and $\bC$ has a $j$-ary primitive positive interpretation $J$ in $\bD$ such that
 $J \circ (I_1,\ldots,I_j)$ is pp-homotopic to the identity interpretation of $\bC$.
Then for every first-order expansion $\bC'$ of $\bC$
there is a first-order expansion $\bD'$ of $\bD$
such that there is are polynomial time reductions between CSP$(\bC')$ and CSP$(\bD')$.
\end{corollary}
\begin{proof}
Lemma~\ref{lem:transfer-new} tells us that for every first-order expansion $\bC'$ of $\bC$
there is a first-order expansion $\bD'$ of $\bD$
such that $\bC'$ and $\bD'$ are mutually pp-interpretable. Thus it is now incumbent on us only to argue that there are polynomial time reductions between CSP$(\bC')$ and CSP$(\bD')$. This essentially generalises the proof of (the hitherto unproved) Lemma~\ref{lem:folklore} and is most easily given in the alternative definition of the CSP as the model-checking problem for primitive positive logic. Let $J'=(j,U,g)$ be a pp-interpretation of $\bC'$ in $\bD'$ where $U$ is given as an $j$-ary pp-formula $\phi_U$ and each $k$-ary relation $R$ of $\bC'$ is given as a $jk$-ary formula $\phi_R$. From an instance of CSP$(\bC')$, that is a pp-sentence in $n$ variables, we build a pp-sentence in $jn$ variables, in which each variable $v$ becomes a sequence $v^1,\ldots,v^j$ that is constrained by $\phi_U$, and where each atom of the form $R(v_1,\ldots,v_k)$ is replaced by $\phi_R(v^1_1,\ldots,v^j_1,\ldots,v^1_k,\ldots,v^j_k)$. This procedure can clearly be made in polynomial time and that it is indeed a reduction follows from it being derived from the corresponding interpretation. The argument for reducing CSP$(\bD')$ to CSP$(\bC')$ is dual, and the result follows.
\end{proof}

\subsection{The Interval Algebra}

We now investigate classification transfer for the Interval Algebra.
\begin{lemma}
There are 2 unary pp-interpretations, $I_1$ and $I_2$, of $(\mathbb{Q};<)$ in $(\mathbb{I},\mathtt{s},\mathtt{f})$, and a binary pp-interpretation $J$ of $(\mathbb{I},\mathtt{s},\mathtt{f})$ in $(\mathbb{Q};<)$, so that $J \circ (I_1,I_2)$ is pp-homotopic to the identity.
\label{lem:IA1}
\end{lemma}
\begin{proof}
Let $I_1$ and $I_2$ be $(1,\top,p_1)$ and $(1,\top,p_2)$, where $p_1$ and $p_2$ are the binary projections to start and finish point of the interval, respectively. Let $J$ be $(2,X^-< X^+,h)$ where $h(X^-,X^+)=[X^-,X^+]$ maps pairs of points $X^- < X^+$ to the interval they specify in the line. 

Define $(X<Y)_{I_1}$ and $(X<Y)_{I_2}$ by $\exists Y',W \; \big(Y (\sa) Y' \wedge X (\sa) W \wedge Y' (\fa) W \big)$ and $\exists X',W \; \big(X (\fa) X' \wedge Y (\fa) W \wedge X' (\sa) W \big)$, respectively.
Define $((u_1,u_2)(\sa)(v_1,v_2))_J$ by $u_1=v_1 \wedge u_2<v_2$ and $((u_1,u_2)(\fa)(v_1,v_2))_J$ by $u_1<v_1 \wedge u_2=v_2$. 

Define $(X=Y)_{I_1}$ as $X(\sa)Y$ and  $(X=Y)_{I_2}$ as $X(\fa)Y$. Define $((u_1,u_2)(\equiv)(v_1,v_2))_J$ as $u_1=v_1 \wedge u_2=v_2$. Now $J \circ (I_1,I_2)$ maps $([X^-,X^+],[Y^-,Y^+])$ to $[Z^-,Z^+]:=[X^-,Y^+]$ and pp-homotopy is given by $X(\sa)Z \wedge Y(\fa)Z$.
%Define $x<y$ by $\exists w \; \big(\mathtt{s}(x,w) \wedge \mathtt{f}(y,w)\big)$
\end{proof}
Let us make the simple observation that all the basic relations of $\mathcal{IA}$ are pp-definable with $<$. Now we can derive the following from Corollary~\ref{cor:transfer-new}.
\begin{corollary}
Let $\Gamma$ be first-order definable in $\mathcal{IA}$ containing the basic relations $\sa$ and $\fa$. Then either CSP$(\Gamma)$ is in P or it is NP-complete.
\label{cor:IA2}
\end{corollary}
Noting that $\sa$ and $\fa$ are pp-definable in $\ma$, though the converse is false, we note that we could have derived the previous corollary with $\ma$ in place of $\sa$ and $\fa$. In that case, it would have been possible to use binary and unary interpretations, $J$ and $I$, respectively, so that $J \circ (I,I)$ would be pp-homotopic to the identity interpretation. For our stronger result, with $\sa$ and $\fa$, such a composition seems unlikely.

Corollary~\ref{cor:IA2} should be compared to the main result in \cite{KrokhinAllen}. They are formally incomparable, since Corollary \ref{cor:IA2} requires certain basic relations to be present, while \cite{KrokhinAllen} only considers binary first-order definable relations.

\subsection{The Rectangle Algebra}

We now investigate classification transfer for the Rectangle Algebra.
We will use the obvious fact that certain relations are pp-definable from the basic relations. Let $(\mathtt{s},\top):=(\mathtt{s},\mathtt{p}) \vee \ldots \vee (\mathtt{s},\equiv)$ be the disjunction with all 13 basic relations relations of the IA in the second coordinate. Then $X(\sa,\top)Y$ iff $\exists Z \left( X(\sa,\pa)Z \wedge Z(\sa,\pu)Y \right)$. Note also that, e.g., $(\sa,\top)$ is pp-definable in $\mathcal{RA}$ by $\exists Z \left( X  (\equiv,\da) W \wedge Y (\equiv,\da) W \wedge X (\sa,\equiv) W\right)$.
\begin{lemma}
There are 4 unary pp-interpretations, $I_1,I_2,I_3$ and $I_4$, of $(\Q;<)$ in $\mathcal{RA}$, and a $4$-ary pp-interpretation $J$ of $\mathcal{RA}$ in $(\Q;<)$, so that $J \circ (I_1,I_2,I_3,I_4)$ is pp-homotopic to the identity.
\label{lem:RA-in-Q}
\end{lemma}
\begin{proof}
For $i \in [4]$, let $I_i$ be $(1,\top,p_i)$, where $p_i$ is the $4$-ary projection to the $i$th coordinate. Let $J$ be $(4,X^-<X^+\wedge Y^-<Y^+,h)$ where $h(X^-,X^+,Y^-,Y^+)=r$ maps two pairs of intervals to the rectangle they specify in the plane.

Define  $(X<Y)_{I_i}$ as follows.

\medskip
\hspace{-0.5cm}
\resizebox{!}{0.85cm}{
$
\begin{array}{l}
(X<Y)_{I_1} \mbox{ by }  \exists Y',W \big(Y (\sa,\top) Y' \wedge X (\sa,\top) W \wedge Y' (\fa,\top) W \big) \\
(X<Y)_{I_2} \mbox{ by }  \exists X',W \big(X (\fa,\top) X' \wedge Y (\fa,\top) W \wedge X' (\sa,\top) W \big)
 \\
(X<Y)_{I_1} \mbox{ by }  \exists Y',W \big(Y (\top,\sa) Y' \wedge X (\top,\sa) W \wedge Y' (\top,\fa) W \big) \\
(X<Y)_{I_2} \mbox{ by }  \exists X',W \big(X (\top,\fa) X' \wedge Y (\top,\fa) W \wedge X' (\top,\sa) W \big)
 \\
\end{array}
$
}
\medskip

\noindent Define $((u_1,u_2,u_3,u_4)(\sa,\top)(v_1,v_2,v_3,v_4))_J$ by $u_1=v_1 \wedge u_2<v_2$ and $((u_1,u_2,u_3,u_4)(\fa,\top)(v_1,v_2,v_3,v_4))_J$ by $u_1<v_1 \wedge u_2=v_2$. The other relations of $\mathcal{RA}$ are defined in the obvious fashion.
%$(\top,\sa)$ and $(\top,\fa)$ are defined in a similar fashion, as are all the other basic relations of $\mathcal{RA}$. 

Define $(X=Y)_{I_1}$, $(X=Y)_{I_2}$, $(X=Y)_{I_3}$ and $(X=Y)_{I_4}$ as $X(\sa,\top)Y$, $X(\fa,\top)Y$, $X(\top,\sa)Y$ and $X(\top,\fa)Y$, respectively. Define $((u_1,u_2,u_3,u_4)(\equiv)(v_1,v_2,v_3,v_4))_J$ as $u_1=v_1 \wedge u_2=v_2 \wedge u_3=v_3 \wedge u_4=v_4$. Now, for $i\in \{1,2,3,4\}$, let $W_i=(X_i,Y_i)$ where $X_i$ and $Y_i$ are intervals and not rectangles. $J \circ (I_1,I_2,I_3,I_4)$ maps 
\[ 
\begin{array}{l}
([X_1^-,X_1^+],[Y_1^-,Y_1^+]),([X_2^-,X_2^+],[Y_2^-,Y_2^+]), \\
\ \ \ \ \ ([X_3^-,X_3^+],[Y_3^-,Y_3^+]),([X_4^-,X_4^+],[Y_4^-,Y_4^+])
\end{array}
\]
to
$Z:=([X_1^-,X_2^+],[Y_3^-,Y_4^+])$ and pp-homotopy is given by $W_1(\sa,\top)Z \wedge W_2(\fa,\top)Z \wedge W_3(\top,\sa)Z \wedge W_4(\top,\fa)Z$.
\end{proof}
Now we can derive the following from Corollary~\ref{cor:transfer-new}.
\begin{corollary}
Let $\Gamma$ be first-order definable in $\mathcal{RA}$ containing the basic relations  $(\sa,\pa)$, $(\sa,\pu)$, $(\fa,\pa)$, $(\fa,\pu)$, $(\pa,\sa)$, $(\pu,\sa)$, $(\pa,\fa)$ and $(\pu,\fa)$. Then either CSP$(\Gamma)$ is in P or it is NP-complete.
\label{cor:RA}
\end{corollary}

\section{Maximal tractability and ORD-Horn}
\label{sec:ORD-Horn}

A set of relations $\Gamma$, over the same domain, is described as \emph{maximally tractable} among $\Delta \supseteq \Gamma$, if every finite subset of $\Gamma$ gives a structure $\bB$ so that CSP$(\bB)$ is in P; while for each $R \in \Delta\setminus \Gamma$, there is a finite subset of $\Gamma$ given by $\bB$, so that CSP$(\bB;R)$ is NP-hard.
%In \cite{BalbianiCondottaCerro}, the maximal tractability of the so-called strongly pre-convex relations, among $\mathcal{RA}^{\vee =}$, is left open. 

\subsection{The Rectangle Algebra} 

In \cite{BalbianiCondottaCerro-journal}, the strongly pre-convex relations are tied precisely to ORD-Horn, yet maximal tractability of this class, among $\mathcal{RA}^{\vee =}$, remains open. Here we settle the question by proving that ORD-Horn is indeed maximally tractable in the Rectangle Algebra. We will do this by using maximal tractability among languages first-order definable in $(\mathbb{Q};<)$ (where in fact ORD-Horn is not maximally tractable).

A formula is called \emph{ll-Horn} if it is a conjunction of formulas
of the following form %(slightly abusing terminology, we call 
%these formulas the \emph{clauses} of the ll-Horn formula)
\begin{align*}
(x_1 = y_1 \wedge \dots \wedge x_k = y_k)  \Rightarrow & (z_1 < z_0\vee \dots \vee z_l < z_0)
\; \hbox{, or}\\
(x_1 = y_1 \wedge \dots \wedge x_k = y_k)  \Rightarrow & (z_1 < z_0 \vee \dots \vee z_l < z_0 \vee \\ 
& (z_0 = z_1 = \dots = z_l)) 
\end{align*}
where $0 \leq k,l$. ORD-Horn is the subclass in which we insist at most a single atom appears in each sequent. Note that $k$ or $l$ might be $0$: if $k=0$, we obtain
a formula of the form $z_1 < z_0 \vee \dots \vee z_l < z_0$ or $(z_1 < z_0 \vee \dots \vee z_l < z_0 \vee (z_0 = z_1 = \dots = z_l))$,
and if $l=0$ we obtain a disjunction of disequalities.
Also note
that the variables $x_1,\dots,x_k,y_1,\dots,y_k,z_0,$
$\dots,z_l$ need not be pairwise distinct. 
On the other hand, the clause $z_1 < z_2 \vee z_3 < z_4$ is an example of a formula that is \emph{not} ll-Horn. The dual class, \dualll-Horn, can be defined as \ll-Horn, but with all $<$ replaced by $>$. The following result is from \cite{tcsps-journal} using the characterisation of \cite{ll}.
\begin{lemma}
The class of relations definable in \ll-Horn (resp., \dualll-Horn) is a maximally tractable subclass of relations fo-definable in $(\Q;<)$.
\end{lemma}
\begin{lemma}
The relation $x=y \Rightarrow u=v$ sits in precisely two maximally tractable classes of constraint languages with respect to relations fo-definable in $(\Q;<)$, which are those whose relations are definable in \ll-Horn and those  whose relations are definable in \dualll-Horn.
\label{lem:little}
\end{lemma}
\begin{proof}
Here we can not avoid some parlance from \cite{tcsps-journal}. Let $\pp$ (resp., $\dualpp$) be the binary operation on $\Q$ that maps $(x,y)$ to $x$, if $x<0$, and $y$, otherwise (resp., maps $(x,y)$ to $y$, if $x<0$, and $x$, otherwise). A relation is \emph{violated} by an operation if, when the operation is applied component-wise to some tuples in the relation, one can obtain a tuple that is not in the relation. Reading from Figure~9 in  \cite{tcsps-journal}, one sees that the present lemma follows if we can prove that $\pp$ and $\dualpp$ both \emph{violate} $x=y \Rightarrow u=v$. To see this for $\pp$, consider the tuples $(-1,-1,2,2)$ and $(1,2,1,2)$ for which $\pp$ produces the tuple $(-1,-1,1,2)$ which is not in the relation. The case $\dualpp$ is achingly similar.
\end{proof}
\noindent Call a definition in \ll-Horn \emph{minimal} if all of its clauses are needed and none can be replaced by one with a smaller sequent on the right-hand side of an implication. We may thus assume that the $z_0,\ldots,z_l$ as in the definition are distinct in each clause. Note that relations definable in ORD-Horn over $(\Q;<)$ is a strict subset of both those relations definable in \ll-Horn and those  relations definable in \dualll-Horn.
\begin{theorem}
The class ORD-Horn on the Rectangle Algebra is maximally tractable among the binary relations fo-definable in $\mathcal{RA}$. %Moreover, ORD-Horn is the unique maximally tractable class containing the basic relations.
\label{prop:RA}
\end{theorem}
\begin{proof}
We know from \cite{BalbianiCondottaCerro-journal} that the ORD-Horn relations among $\mathcal{RA}^{\vee =}$ give a CSP that is tractable. Let $R$ be some binary relation not definable in ORD-Horn and let $J$ be as in Lemma~\ref{lem:RA-in-Q}. Suppose we translate $R$ to an $8$-ary relation $S$ over $(\Q;<)$ via $J$, and let us make the similar translation for all ORD-Horn relations of the Rectangle Algebra, which will become the set of relations $\Gamma$ fo-definable in $(\mathbb{Q};<)$. Owing to Lemma~\ref{lem:RA-in-Q}, we need only argue that CSP$(\Q;\Gamma,S)$ is NP-hard (we abuse notation by presuming $\Gamma$ defines also a set of relations over $\mathbb{I}$).

Since the (ORD-Horn) relation $\{s,\equiv,s^\smile\} \times (\top \setminus \{s,\equiv,s^\smile\}$ translates to $X^-=U^- \Rightarrow Y^-=V^-$ in $\Gamma$, we can deduce from Lemma~\ref{lem:little} that the only maximally tractable classes for relations fo-definable in $(\Q;<)$ that $\{S\} \cup \Gamma$ can sit in are those corresponding with $\ll$ and $\dualll$. In particular, if $S$ is outside of these classes then it follows immediately that CSP$(\Q;\Gamma,S)$ is NP-hard. Let us therefore assume \mbox{w.l.o.g.} that $S$ is within $\ll$, as the other case is dual, and we will seek a contradiction.

Let $\phi$ be some minimal \ll-Horn specification of $S$. 
%Note that $\phi$ arising from can not contain any atoms $x<y$ where $x$ and $y$ are endpoints of intervals associated with distinct dimensions. To see this, one can use the automorphisms of any relations fo-definable in $\mathbb{RA}$ which translate one dimension while leaving the other unchanged.
Consider some clause that is not ORD-Horn. It involves some sequent of the form either $z_1<z_0 \vee \ldots \vee z_l<z_0$ or  $z_1<z_0 \vee \ldots \vee z_l<z_0 \vee z_0=z_1=\ldots=z_l$ where $z_0,z_1,\ldots,z_l$ are distinct variables. Now, since $R$ was binary, we know that some rectangle  is mentioned twice among the $z_0,z_1,\ldots,z_l$. If some comparison involves (w.l.o.g.) $I^-$ and $I^+$, then we can remove this and we contradict minimality (note that $X^-<X^+$ because we do not allow point-like intervals). Thus, if we are not already ORD-Horn we must have something of the form, again w.l.o.g., $X^-<Y^- \vee X^-<Y^+\vee \ldots$ in the sequent, but this can be simplified to $X^-<Y^+ \vee \ldots$ contradicting minimality. Thus the assumption that we be minimal actually makes us ORD-Horn.
\end{proof}

\subsection{The $r$-dimensional Block Algebra}

Firstly, we will profit from studying certain automorphisms of the Block Algebra. An \emph{automorphism} of a structure $\bB$ is a permutation $f$ on its domain so that, for all relations $R$ of $\bB$, of arity $k$, and all $k$-tuples of domain elements $x_1,\ldots,x_k \in B$, $R(x_1,\ldots,x_k)$ iff $R(f(x_1),\ldots,f(x_k))$. The Interval Algebra $\mathcal{IA}$ enjoys all \emph{translation} automorphisms, of the form $[X^-,X^+] \mapsto [q+X^-,q+X^+]$, for any $q \in \Q$. The Block Algebra $\mathcal{BA}_r$ enjoys these similarly, independently for each of its axes. That is, $\mathcal{BA}_r$ has all automorphisms of the form $([X_1^-,X_1^+],\ldots,[X_r^-,X_r^+]) \mapsto ([q_1+X_1^-,q_1+X_1^+],\ldots,[q_r+X_r^-,q_r+X_r^+])$. In particular, there is an automorphism of $\mathcal{BA}_r$ that translates one axis any amount, while leaving the other axes where they are. An important property of automorphisms that we will use is that the truth of a first-order formula is invariant under an automorphism. That is, if $\phi(x_1,\ldots,x_k)$ is a first-order formula over  $\mathcal{BA}_r$, and $h$ is an automorphism of $\mathcal{BA}_r$, then $\mathcal{BA}_r \models \phi(x_1,\ldots,x_k)$ iff $\mathcal{BA}_r \models \phi(h(x_1),\ldots,h(x_k))$.

A \emph{hypercuboid} is a polytope specified in $k$-dimensional space by the intersection of intervals $c_i\leq x_i \leq d_i$, for $c_i,d_i \in \Q$ and $i \in \{1,\ldots,k\}$. We are now in a position to extend our classification transfers to the $r$-dimensional Block Algebra.
\begin{theorem}
The class ORD-Horn on the $r$-dimensional Block Algebra is maximally tractable with respect to the binary relations fo-definable in $\mathcal{BA}_r$.
\label{prop:BA}
\end{theorem}
\begin{proof}
We follow the proof of the Theorem~\ref{prop:RA} up to the point where we consider some clause that is not ORD-Horn. It involves some sequent of the form either $z_1<z_0 \vee \ldots \vee z_l<z_0$ or  $z_1<z_0 \vee \ldots \vee z_l<z_0 \vee z_0=z_1=\ldots=z_l$ where $z_0,z_1,\ldots,z_l$ are distinct variables. Now, since $R$ was binary, we know that some hypercuboid is mentioned twice among the $z_0,z_1,\ldots,z_l$.

Case A. The same dimension is mentioned twice.  This has been dealt with in the proof of Theorem~\ref{prop:RA}.

Case B. No dimension is mentioned twice, but we have an atom of the form $X^p<Y^q$, for $p,q \in \{+,-\}$, where $I$ and $J$ are different dimensions of the same hypercuboid. The $r$-dimensional Block Algebra enjoys an automorphism that translates the dimension associated with $X$ while leaving unchanged all the other dimensions. Consider the evaluation that witnesses that the atom $X^p<Y^q$ is true whilst everything else in that clause is false. Now applying an automorphism we can falsify this atom while leaving all the others of that clause false (remember $X^{p}$ is not repeated and $X^{\overline{p}}$, where $\overline{p} \in \{-,+\}\setminus\{p\}$, does not appear anywhere since then we would be in Case A). This demonstrates that $\phi$ does not specify a correct translation from the $r$-dimensional Block Algebra (the truth of $\phi$ should be invariant under these automorphisms).

Thus the assumption that we be minimal actually makes us ORD-Horn.
\end{proof}
Note that our proof that ORD-Horn is maximally tractable works also for the Interval Algebra, where that result famously originated in \cite{Nebel}. The only change we need to make is with Lemma~\ref{lem:little}, because $x=y \Rightarrow u=v$ can not originate from the Interval Algebra, but $x=y \Rightarrow u=v \wedge y>x \wedge v>u$ can, and suffices for our argument. Note that the analog of Corollary~\ref{cor:RA} for the Block Algebra $\mathcal{BA}_{r}$ proceeds with almost identical proof.

\section{Final Remarks}
\label{sec:final}

Our approach is not tailored to the Rectangle or Block Algebras.
Indeed, its motivation is inextricably linked to its
wide scope. We are currently working on deriving full complexity
classifications, above the basic relations, in other formalisms,
for example the Cardinal Direction Calculus\footnote{Not to be confused with the Cardinal Direction Relations, see \cite{TemporalSpatialSurvey} for disambiguation.} of \cite{Frank91} and the Directed Interval Algebra of \cite{Renz}. Applications to these settings will appear in the long version of this article.

\bibliographystyle{named}

%\bibliography{local}

\begin{thebibliography}{}

\bibitem[\protect\citeauthoryear{Ahlbrandt and
  Ziegler}{1986}]{AhlbrandtZiegler}
Gisela Ahlbrandt and Martin Ziegler.
\newblock Quasi finitely axiomatizable totally categorical theories.
\newblock {\em Annals of Pure and Applied Logic}, 30(1):63 -- 82, 1986.

\bibitem[\protect\citeauthoryear{Allen}{1983}]{Allen}
J.~F. Allen.
\newblock Maintaining knowledge about temporal intervals.
\newblock {\em Communications of the ACM}, 26(11):832--843, 1983.

\bibitem[\protect\citeauthoryear{Balbiani \bgroup \em et al.\egroup
  }{2002}]{BalbianiCondottaCerro-journal}
Philippe Balbiani, Jean‐François Condotta, and Luis Fariñas~Del Cerro.
\newblock Tractability results in the block algebra.
\newblock {\em Journal of Logic and Computation}, 12(5):885--909, 2002.

\bibitem[\protect\citeauthoryear{Bodirsky and Grohe}{2008}]{BodirskyGrohe}
Manuel Bodirsky and Martin Grohe.
\newblock Non-dichotomies in constraint satisfaction complexity.
\newblock In {\em Automata, Languages, and Programming - 35th International
  Colloquium, {ICALP} 2008}, pages 184--196, 2008.

\bibitem[\protect\citeauthoryear{Bodirsky and K{\'{a}}ra}{2010}]{tcsps-journal}
Manuel Bodirsky and Jan K{\'{a}}ra.
\newblock The complexity of temporal constraint satisfaction problems.
\newblock {\em J. {ACM}}, 57(2), 2010.

\bibitem[\protect\citeauthoryear{Bodirsky and K\'ara}{2015}]{ll}
Manuel Bodirsky and Jan K\'ara.
\newblock A fast algorithm and lower bound for temporal reasoning.
\newblock {\em ACM Trans. Comput. Logic}, 62(3):19:1--19:52, 2015.

\bibitem[\protect\citeauthoryear{Bodirsky and Pinsker}{2015}]{BodPin-Schaefer}
Manuel Bodirsky and Michael Pinsker.
\newblock Schaefer's theorem for graphs.
\newblock {\em J. ACM}, 62(3):19:1--19:52, 2015.

\bibitem[\protect\citeauthoryear{Bodirsky and W{\"{o}}lfl}{2011}]{BodirskyW11}
Manuel Bodirsky and Stefan W{\"{o}}lfl.
\newblock {RCC8} is polynomial on networks of bounded treewidth.
\newblock In {\em {IJCAI} 2011, Proceedings of the 22nd International Joint
  Conference on Artificial Intelligence, Barcelona, Catalonia, Spain, July
  16-22, 2011}, pages 756--761, 2011.

\bibitem[\protect\citeauthoryear{Bodirsky \bgroup \em et al.\egroup
  }{2012}]{HornOrFull}
Manuel Bodirsky, Peter Jonsson, and Timo von Oertzen.
\newblock Horn versus full first-order: a complexity dichotomy for algebraic
  constraint satisfaction problems.
\newblock {\em J. Logic and Comput.}, 22(3):643--660, 2012.

\bibitem[\protect\citeauthoryear{Bodirsky \bgroup \em et al.\egroup
  }{2015}]{dCSPs3}
Manuel Bodirsky, Barnaby Martin, and Antoine Mottet.
\newblock Discrete temporal constraint satisfaction problems.
\newblock {\em CoRR}, abs/1503.08572, 2015.
\newblock Accepted for publication in {J. ACM}.

\bibitem[\protect\citeauthoryear{Bulatov}{2017}]{FVproofBulatov}
Andrei~A. Bulatov.
\newblock A dichotomy theorem for nonuniform csps.
\newblock {\em CoRR}, abs/1703.03021, 2017.
\newblock Extended abstract appeared at The 58th Annual Symposium on
  Foundations of Computer Science (FOCS 2017).

\bibitem[\protect\citeauthoryear{Dowling and Gallier}{1984}]{HornSATLinear}
William~F. Dowling and Jean~H. Gallier.
\newblock Linear-time algorithms for testing the satisfiability of
  propositional {H}orn formulae.
\newblock {\em The Journal of Logic Programming}, 1(3):267 -- 284, 1984.

\bibitem[\protect\citeauthoryear{Dylla \bgroup \em et al.\egroup
  }{2017}]{TemporalSpatialSurvey}
Frank Dylla, Jae~Hee Lee, Till Mossakowski, Thomas Schneider, Andr{\'{e}} van
  Delden, Jasper van~de Ven, and Diedrich Wolter.
\newblock A survey of qualitative spatial and temporal calculi: Algebraic and
  computational properties.
\newblock {\em {ACM} Comput. Surv.}, 50(1):7:1--7:39, 2017.

\bibitem[\protect\citeauthoryear{Feder and Vardi}{1999}]{FederVardi}
T.~Feder and M.~Vardi.
\newblock The computational structure of monotone monadic {SNP} and constraint
  satisfaction: {A} study through {D}atalog and group theory.
\newblock {\em {SIAM} Journal on Computing}, 28:57--104, 1999.

\bibitem[\protect\citeauthoryear{Frank}{1991}]{Frank91}
A.~U. Frank.
\newblock Qualitative spatial reasoning with cardinal directions.
\newblock In {\em {\"OGAI} (Informatik Fachberichte)}, page 157–167, 1991.
\newblock Vol. 287. Springer.

\bibitem[\protect\citeauthoryear{Guesgen}{1989}]{Giisgen89}
Hans~Werner Guesgen.
\newblock Spatial reasoning based on allen's temporal logic.
\newblock Technical report, International Computer Science Institute, 1989.

\bibitem[\protect\citeauthoryear{Hirsch}{1996}]{Hirsch}
R.~Hirsch.
\newblock Relation algebras of intervals.
\newblock {\em Artificial Intelligence Journal}, 83:1--29, 1996.

\bibitem[\protect\citeauthoryear{Horn}{1951}]{horn1951}
Alfred Horn.
\newblock On sentences which are true of direct unions of algebras.
\newblock {\em Journal of Symbolic Logic}, 16(1):14--21, 1951.

\bibitem[\protect\citeauthoryear{Jeavons}{1998}]{Jeavons}
P.~G. Jeavons.
\newblock On the algebraic structure of combinatorial problems.
\newblock {\em Theoretical Computer Science}, 200:185--204, 1998.

\bibitem[\protect\citeauthoryear{Jonsson and Drakengren}{1997}]{RCC5JD}
Peter Jonsson and Thomas Drakengren.
\newblock A complete classification of tractability in {RCC}-5.
\newblock {\em J. Artif. Intell. Res.}, 6:211--221, 1997.

\bibitem[\protect\citeauthoryear{Krokhin \bgroup \em et al.\egroup
  }{2003}]{KrokhinAllen}
Andrei Krokhin, Peter Jeavons, and Peter Jonsson.
\newblock Reasoning about temporal relations: The tractable subalgebras of
  {A}llen's interval algebra.
\newblock {\em J.ACM}, 50(5):591--640, 2003.

\bibitem[\protect\citeauthoryear{Ligozat}{1998}]{LigozatCDC}
G.~É. Ligozat.
\newblock Reasoning about cardinal directions.
\newblock {\em Journal of Visual Languages \& Computing}, 9(1):23 -- 44, 1998.

\bibitem[\protect\citeauthoryear{Mukerjee and Joe}{1990}]{MukerjeeJ90}
Amitabha Mukerjee and Gene Joe.
\newblock A qualitative model for space.
\newblock In {\em Proceedings of the 8th National Conference on Artificial
  Intelligence. Boston, Massachusetts, July 29 - August 3, 1990, 2 Volumes.},
  pages 721--727, 1990.

\bibitem[\protect\citeauthoryear{Nebel and B\"urckert}{1995}]{Nebel}
Bernhard Nebel and Hans-J\"urgen B\"urckert.
\newblock Reasoning about temporal relations: A maximal tractable subclass of
  {A}llen's interval algebra.
\newblock {\em J.ACM}, 42(1):43--66, 1995.

\bibitem[\protect\citeauthoryear{Nebel}{1995}]{NebelFirst}
Bernhard Nebel.
\newblock Computational properties of qualitative spatial reasoning: First
  results.
\newblock In {\em {KI-95:} Advances in Artificial Intelligence, 19th Annual
  German Conference on Artificial Intelligence, Bielefeld, Germany, September
  11-13, 1995, Proceedings}, pages 233--244, 1995.

\bibitem[\protect\citeauthoryear{Randell \bgroup \em et al.\egroup
  }{1992}]{Randell92aspatial}
David~A. Randell, Zhan Cui, and Anthony~G. Cohn.
\newblock A spatial logic based on regions and connection.
\newblock In {\em Proceedings 3rd International Conference on Knowledge
  Representation and Reasoning}, 1992.

\bibitem[\protect\citeauthoryear{Reingold}{2008}]{RheingoldJACM}
Omer Reingold.
\newblock Undirected connectivity in log-space.
\newblock {\em J. ACM}, 55(4):1--24, 2008.

\bibitem[\protect\citeauthoryear{Renz and Nebel}{2007}]{NebelRenzSurvey}
J.~Renz and B.~Nebel.
\newblock Qualitative spatial reasoning using constraint calculi.
\newblock In M.~Aiello, I.~Pratt-Hartmann, and J.~van Benthem, editors, {\em
  Handbook of Spatial Logics}. Springer Verlag, Berlin, 2007.

\bibitem[\protect\citeauthoryear{Renz}{2001}]{Renz}
Jochen Renz.
\newblock A spatial odyssey of the interval algebra: 1. directed intervals.
\newblock In {\em Proceedings of the 17th International Joint Conference on
  Artificial Intelligence - Volume 1}, IJCAI'01, pages 51--56, San Francisco,
  CA, USA, 2001. Morgan Kaufmann Publishers Inc.

\bibitem[\protect\citeauthoryear{Vilain and Kautz}{1986}]{VilainK86}
Marc~B. Vilain and Henry~A. Kautz.
\newblock Constraint propagation algorithms for temporal reasoning.
\newblock In {\em Proceedings of the 5th National Conference on Artificial
  Intelligence. Philadelphia, PA, August 11-15, 1986. Volume 1: Science.},
  pages 377--382, 1986.

\bibitem[\protect\citeauthoryear{Zhuk}{2017}]{FVproofZhuk}
Dmitriy Zhuk.
\newblock The proof of {CSP} dichotomy conjecture.
\newblock {\em CoRR}, abs/1704.01914, 2017.
\newblock Extended abstract appeared at The 58th Annual Symposium on
  Foundations of Computer Science (FOCS 2017).

\end{thebibliography}

\section*{Appendix}

Let us additionally note that using the famous result of \cite{RheingoldJACM}, one may improve Lemma~\ref{lem:folklore} and Corollary~\ref{cor:transfer-new} to logspace reductions.
 
\subsection*{The Cardinal Direction Calculus}

Recall the Cardinal Direction Calculus of \cite{Frank91}. Let $\mathcal{CDC}$ be the structure whose domain elements are rational points in the plane with basic relations $\mathtt{N},\mathtt{S},\mathtt{E},\mathtt{W},\mathtt{NE},\mathtt{SE},\mathtt{SW},\mathtt{NW}$. These are interpreted projectively, in the following sense. $x (\mathtt{N}) y$ (read as ``$x$ north of $y$'') holds iff $x$ and $y$ have the same position on the horizontal axis and $x$ is above $y$ on the vertical axis. Then, $x (\mathtt{NE}) y$ iff $\exists y' (y' (\mathtt{E}) y \wedge x (\mathtt{N}) y')$ (read $x (\mathtt{NE}) y$ as ``$x$ is northeast of $y$''). The other relations can be read from compass points in the obvious fashion.
\begin{lemma}
There are two unary pp-interpretations $I_1$ and $I_2$ of $(\mathbb{Q};<)$ in $\mathcal{CDC}$, and a binary pp-interpretation $J$ of $\mathcal{CDC}$ in $(\mathbb{Q};<)$, so that $J \circ (I_1,I_2)$ is pp-homotopic to the identity interpretation.
\end{lemma}
\begin{proof}
Let $I_1$ and $I_2$ be the interpretations $(1,\top,p_1)$ and $(1,\top,p_2$) where $p_1$ and $p_2$ are the binary projections to first and second coordinate, respectively. Let $J$ be $(2,\top,h)$ where $h(x,y)=(x,y)$ maps pairs of rationals to the point they specify in the plane. $(X<Y)_{I_1}$ is $\exists X',Y' \, (X' (\mathtt{N}) X \wedge Y' (\mathtt{N}) Y \wedge X' (\mathtt{W}) Y')$ and $(X<Y)_{I_2}$ is $\exists X',Y' \, (X' (\mathtt{E}) X \wedge Y' (\mathtt{E}) Y \wedge X' (\mathtt{S}) Y')$. 

Each basic relation of $\mathcal{CDC}$ is pp-definable from the basic relation $<$ on $(x,y)$ and $(x',y')$ as follows.
\[
\begin{array}{ll}
\mathtt{N} \ \ \ \ \ \ \ \ & y>y' \wedge x=x' \\
\mathtt{E} & x>x' \wedge y=y' \\
\mathtt{S} & y'>y \wedge x=x' \\
\mathtt{W} & x'>x \wedge y=y' \\
\mathtt{NE} & y>y' \wedge x>x' \\
\mathtt{SE} & y'>y \wedge x>x' \\
\mathtt{SW} & y'>y \wedge x'>x \\
\mathtt{NW} & y>y' \wedge x'>x \\
\end{array}
\]
$(X=Y)_{I_1}$ is $\exists Z \, (Z (\mathtt{N})X \wedge Z (\mathtt{N})Y)$ and $(X=Y)_{I_2}$ is $\exists Z \, (Z (\mathtt{E})X \wedge Z (\mathtt{E})Y)$. $((u_1,u_2)(\equiv)(v_1,v_2))_J$ is $u_1=v_1 \wedge u_2=v_2$. Now $J \circ (I_1,I_2)$ maps $(U,V)$ = $((u_1,u_2),(v_1,v_2)))$ to $Z:=(u_1,v_2)$ and pp-homotopy is given by $\exists U'V' \, (U'(\mathtt{N})U \wedge U'(\mathtt{N})Z \wedge V'(\mathtt{E})V \wedge V'(\mathtt{E})Z)$.
\end{proof}
%Note that this argument works for both the projective and cone interpretations of $\mathtt{N},\mathtt{S},\mathtt{E},\mathtt{W},\mathtt{NE},\mathtt{SE},\mathtt{SW},\mathtt{NW}$.

\begin{corollary}
Let $\Gamma$ be a first-order expansion of $\mathcal{CDC}$. Either CSP$(\Gamma)$ is in P or it is NP-complete.
\end{corollary}

\subsection{Directed Interval Algebra}

Recall the Directed Interval Algebra of \cite{Renz}, and in particular the two binary relations $\mathtt{cb}$ and $\mathtt{cf}$, which correspond to $\sa$ and $\fa$ in the Interval Algebra. They are annotated with a subscript $=$ to indicate the directions of the two intervals are concomitant. Let $\mathcal{DIA}$ be the directed interval algebra, with its basic relations only, on the line of rationals and let $\mathcal{DIA}^{\rightarrow}$ be this structure augmented with a unary relation indicating the direction is coincident with that of the axis (that is, forwards).
\begin{lemma}
There are two unary pp-interpretations, $I_1$ and $I_2$, of $(\mathbb{Q};<)$ in $\mathcal{DIA}^{\rightarrow}$, and a binary pp-interpretation $J$ of $\mathcal{DIA}^{\rightarrow}$ in $(\mathbb{Q};<)$, so that $J \circ (I_1,I_2)$ is pp-homotopic with the identity interpetation.
\end{lemma}
\begin{proof}
Let $\mathit{forw}$ be the unary relation of $\mathcal{DIA}^{\rightarrow}$ denoting ``forwards''.
Let $I_1$ and $I_2$ be $(1,\mathit{forw}(X),p_1)$ and $(1,\mathit{forw}(X),p_2)$, where $p_1$ and $p_2$ are the binary projections to start and finish point of the interval, respectively (following the direction of the arrow). Let $J$ be $(2,X^-\leq X^+,h)$ where $h(X^-,X^+)=X$ maps pairs of points $X^- \leq X^+$ to the interval they specify in the line.

Define $(X<Y)_{I_1}$ and $(X<Y)_{I_2}$ by $\exists Y',W \; \big(Y (\mathtt{cb}_=) Y' \wedge X (\mathtt{cb}_=) W \wedge Y' (\mathtt{cf}_=) W \big)$ and $\exists X',W \; \big(X (\mathtt{cf}_=) X' \wedge Y (\mathtt{cf}_=) W \wedge X' (\mathtt{cb}_=) W \big)$, respectively.
Define $((u_1,u_2)(\mathtt{cb}_=)(v_1,v_2))_J$ by $u_1=v_1 \wedge u_2<v_2$ and $((u_1,u_2)(\mathtt{cf}_=)(v_1,v_2))_J$ by $u_1<v_1 \wedge u_2=v_2$. 

Define $(X=Y)_{I_1}$ as $X(\mathtt{cb}_=)Y$ and  $(X=Y)_{I_2}$ as $X(\mathtt{cf}_=)Y$. Define $((u_1,u_2)(\equiv)(v_1,v_2))_J$ as $u_1=v_1 \wedge u_2=v_2$. Now $J \circ (I_1,I_2)$ maps $([X^-,X^+],[Y^-,Y^+])$ to $[Z^-,Z^+]:=[X^-,Y^+]$ and pp-homotopy is given by $X(\mathtt{cb}_=)Z \wedge Y(\mathtt{cf}_=)Z$.
\end{proof}
Note that we can not pp-define $<$ in the basic relations of $\mathcal{DIA}$ without having restricted ourselves to just directed intervals in $\mathit{forw}$.
\begin{lemma}
For every first-order expansion $\Gamma$ of $\mathcal{DIA}$, there exists a first-order expansion $\Gamma'$ of $\mathcal{DIA}^{\rightarrow}$ so that CSP$(\Gamma)$ and CSP$(\Gamma')$ have polynomially equivalent complexities.
\end{lemma}
\begin{proof}
Let $\Gamma$ be a first-order expansion of $\mathcal{DIA}$ and let $(\Gamma;\mathit{forw})$ be the corresponding first-order expansion of $\mathcal{DIA}^{\rightarrow}$. The reduction from CSP$(\Gamma;\mathit{forw})$ to CSP$(\Gamma)$ is the identity. We consider the reduction from CSP$(\Gamma;\mathit{forw})$ to CSP$(\Gamma)$. Let $\phi$ be an instance of the former. We build $\phi'$ from $\phi$ as follows. Enumerate the constraints in $\phi$ of the form $\mathit{forw}(X)$ as $\mathit{forw}(X_1),\ldots,\mathit{forw}(X_m)$. In $\phi'$ these are replaced by binary relations $X_1(\mathit{same})X_2 \wedge \ldots \wedge X_{m-1}(\mathit{same})X_m)$ where $U(\mathit{same})V$ is defined as $\exists U',V' \, (U(\mathtt{e}_=)U' \wedge V(\mathtt{e}_=)V' \wedge U'(\mathtt{eq}_=)V')$, and forces both directed intervals to be oriented the same way. Clearly this reduction is polynomial, we will now prove it is correct.

Suppose $(\Gamma;\mathit{forw}) \models \phi$. Then $\Gamma \models \phi'$ by choosing the forward direction for everything appearing in the new $\mathit{same}$ constraints.

Suppose $\Gamma \models \phi'$. Then we can say $(\Gamma;\mathit{forw}) \models \phi$ since all of the constraints of $\phi$ are invariant under the automorphism of $\Gamma$ that maps each directed interval to its inverse under $\mathtt{Eq}_{\neq}$. 
\end{proof}
\begin{corollary}
Let $\Gamma$ be a first-order expansion of $\mathcal{DIA}$. Either CSP$(\Gamma)$ is in P or it is NP-complete.
\end{corollary}

\subsection*{Acknowledgements}

We are grateful to our IJCAI reviewers for corrections to the first draft of the paper.

\end{document}